\documentclass{sig-alternate}
\usepackage[usenames,dvipsnames]{color}
\usepackage{supertech,clrscode,xspace,amsfonts,relate,epstopdf,color,defs,graphicx,url,cite,amsmath,epsfig,epsf,latexsym}

\definecolor{light-gray}{gray}{0.93}

\begin{document}
\clubpenalty=10000 
\widowpenalty = 10000

\title{Reallocation Problems in Scheduling}
\numberofauthors{5}

\author{
\alignauthor Michael A.~Bender
 \affaddr{Computer Science,}\\
 \affaddr{\mbox{Stony Brook University}}\\
 \affaddr{\mbox{\ and Tokutek, Inc.}}\\
 \affaddr{USA}\\
\and
\alignauthor Martin Farach-Colton\\ 
 \affaddr{{Computer Science,}}\\
 \affaddr{\mbox{Rutgers University}}\\
 \affaddr{\mbox{\ and Tokutek, Inc.}}\\
 \affaddr{USA}\\
\and
\alignauthor S\'andor P. Fekete\\
 \affaddr{Computer Science}\\
 \affaddr{TU Braunschweig}\\
 \affaddr{\ }\\
\affaddr{Germany}\\
\and
\alignauthor Jeremy T. Fineman\\ 
  \affaddr{Computer Science}\\
 \affaddr{Georgetown University}\\
\affaddr{USA}\\
\and
\alignauthor Seth Gilbert\\
  \affaddr{Computer Science}\\
 \affaddr{\mbox{National University of Singapore}}\\
 \affaddr{Singapore}\\
}

\maketitle
\footnotenonumber{This research was supported in part by NSF grants IIS 1247726,  
IIS 1247750, 
CCF 1114930, 
CCF 1217708, 
CCF 1114809, 
CCF 0937822,  
CCF 1218188, and  
by Singapore NUS FRC R-252-000-443-133. 
}

\terms{Algorithms, optimization.}

\begin{abstract} 
In traditional on-line problems, such as scheduling,
requests arrive over time, demanding available resources.  As each request
arrives, some resources may have to be irrevocably committed to servicing that
request.  In many situations, however, it may be possible or even
necessary to \emph{reallocate}
previously allocated resources in order to satisfy a new request.  This
reallocation has a cost.  This paper shows how to service the requests
while minimizing the reallocation cost. 

We focus on the classic problem of scheduling jobs on a multiprocessor
system.  Each unit-size job has a time window in which it can be
executed.  Jobs are dynamically added and removed from the system. We
provide an algorithm that maintains a valid schedule, as long as a
sufficiently feasible schedule exists.  The algorithm reschedules only
$O(\min\{\log^*{n}, \log^*{\Delta}\})$ jobs for each job that is
inserted or deleted from the system, where $n$ is the number of active
jobs and $\Delta$ is the size of the largest window.

\end{abstract}

\section{Introduction}
\label{intro}

Imagine you are running a doctor's office.  Every day, patients call
and try to schedule an appointment, specifying a time period in which
they are free.  You respond by agreeing to a specific appointment
time.  Sometimes, however, there is no available slot during the
period of time specified by the patient.  What should you do?  You
might simply turn the
patient away.  Or, you can reschedule some of your existing
patients, making room in the schedule.\footnote{Before you get too
  skeptical about the motivation, this is exactly what M.~F-C's
  ophthalmologist does.}  Unfortunately, patients do not
like being rescheduled.  How do you minimize the number of patients
whose appointments are rescheduled?  

While scheduling a doctor's office may (or may not) seem a somewhat
contrived motivating example, this situation arises with
frequency in real-world applications.  Almost any scenario that
involves creating a schedule also requires the flexibility to later
change that schedule, and those changes often have real costs 
(measured in equipment, computation, or tempers).  For example, in
the computational 
world, scheduling jobs on multiprocess machines and scheduling
computation on the cloud lead to rescheduling.  In the
physical world, these problems arise with depressing regularity in
scheduling airports and train stations.
Real schedules are always changing.

In a tightly packed schedule, it can be difficult to
perform this rescheduling efficiently.  Each task you reschedule
risks triggering a cascade of other reschedulings, leading to  high
costs (and unhappy patients).  It is easy to construct an example where
each job added or removed changes $\Omega(n)$ other
jobs, even with constant-sized tasks.  In this paper, we show
that if there is slack in the 
schedule, then these rescheduling cascades can be collapsed, in fact
down to $O(\log^*n)$ for unit-size jobs.  

\subsection*{Reallocation Problems} 
We introduce a framework for
studying the familiar topic of how to change resource allocations as
problem instances change, with a goal of unifying results of this
type, e.g.,~\cite{HallPo04,UnalUzKi97,SandersSiSk09}. We call problems
in this framework \defn{reallocation problems}.  A reallocation
problem is online in the sense that requests arrive and the system
responds. Unlike in the standard online setting where
resources are irrevocably assigned, in a reallocation problem,
allocations may change.  These reallocations, however, have a cost.

Reallocation lies somewhere between traditional notions of offline and
online resource allocation.  If the reallocation cost is $0$, then
there is no penalty for producing an optimal allocation after each
request. In this case, a reallocation problem can be viewed as a
sequence of offline problems. If the cost of reallocation is $\infty$,
then no finite-cost reallocation is possible and the result is a traditional
online problem.  When there is a bounded but non-zero cost for
reallocation, then there is a trade-off between the quality of an
allocation and the cost of reallocation.

Many related questions have been asked in the scheduling community
(explored more fully below), including: how can one design schedules
that are robust to uncertain or noisy inputs (see,
e.g.,~\cite{KouvelisYu97,MulveyVaZe95}); how can one generate
schedules that change in a limited way while still remaining close to
optimal~\cite{Tovey86}; what is the computational cost of finding a
new optimal schedule as the inputs change
(e.g.,~\cite{ArchettiBeSp10,AusielloEsMo09,ArchettiBeSp03,BockenhauerFoHr06}).
Our approach differs in that it is job-centered, meaning that we
measure the cost of moving jobs rather than the cost of computing
where jobs should move to.

Reallocation is a natural problem.  Many existing algorithms, when
looked in the right way, can be viewed as reallocation problems, e.g.,
reconfiguring FPGAs~\cite{FeketeKaSc12}, maintaining a sparse
array~\cite{ItaiKoRo81,Willard82,Willard86,Willard92,BenderHu07}, or maintaining
an on-line topological ordering
(e.g.,~\cite{KatrielBo05,HaeuplerKaMa08,BenderFiGi09}).  We believe
that the framework developed in this paper will allow us to achieve
new insights into classical scheduling and optimization problems and
the cost of changing a good solution when circumstances change.

\subsection*{Our Problem}  
We focus on the reallocation version of a classical multiprocessor
scheduling problem~\cite{Jackson55} (described more fully in
\secref{model}).  We are given a set of unit-length jobs to process on
$m$ machines.  Each job has an arrival time and a deadline.  The job
must be assigned to a machine and processed at some point within the
specified time window.  Jobs are added and removed from the schedule
dynamically.  The goal is to maintain a feasible schedule at all
times.

In order to process a request, it may be necessary to reschedule some
previously scheduled jobs.  There are two ways in which a job may be rescheduled:
it may be \defn{reallocated} to another time on the same machine, or
it may be \defn{migrated} to a different machine.  The \defn{migration
  cost} is the total number of jobs that are moved to different 
machines when new jobs are added or removed.  
The \defn{reallocation cost} is
the total number of jobs that are rescheduled, regardless of 
whether they are migrated or retained on the same machine. 
Our goal is to minimize
both the migration cost and the total reallocation cost. 
We bound these costs separately, since we expect that a 
reallocation might be more expensive if it also entails a migration. 
(See~\cite{BecchettiLeMu04,AwerbuchAzLe02} 
for other work that considers migrations separately 
from other scheduling considerations, such as 
pre\"emptions.)

We call an algorithm that processes such a sequence of scheduling
requests a \defn{reallocating scheduler}.  We show in 
\secref{why} that  a reallocating scheduler must allow for some job
migrations and that  there is no efficient reallocating scheduler
without some form of resource augmentation; here we consider 
speed augmentation~\cite{KalyanasundaramPr95, PhillipsStTo02}.
We say that an instance is
\defn{$\gamma$-underallocated} if it is feasible even when all jobs
sizes (processing times) are multiplied by $\gamma$. In other words,
the offline scheduler is $\gamma$ times slower than the online
scheduler.

\subsection*{Results}
This paper gives an efficient $m$-machine reallocating scheduler for
unit-sized jobs with arrival times and deadlines.  Informally, the
paper shows that as long as there is sufficient slack 
(independent of $m$) in the requested
schedule, then every request is fulfilled, the reallocation cost is
small, and at most one job migrates across machines on each request.
Specifically, this paper establishes the following theorem:

\begin{theorem}
  There exists a constant $\gamma$ as well as a reallocating scheduler 
  for unit-length jobs such
  that for any $m$-machine $\gamma$-underallocated sequence of
  scheduling requests, we achieve the following performance.  Let
  $n_i$ denote the number of jobs in the schedule and $\Delta_i$ the
  largest window size when the $i$th
  reallocation takes place.  Then the $i$th reallocation
\begin{closeitemize}
\item  has cost $O(\min\left\{\log^* n_{i},\log^* \Delta_{i}\right\})$, and 
\item requires at most one machine migration. 
\end{closeitemize}
\thmlabel{unit-size}
\end{theorem}

We prove \thmref{unit-size} in stages. 
In \secreftwo{parallel}{aligned}, we assume that job windows are
all nicely ``aligned,'' 
by which we mean that all job windows are either disjoint, 
or else one is completely contained in the other. 
In \secref{parallel}, we show that the
multi-machine aligned case can be reduced to the single-machine
aligned case, sacrificing a constant-factor in the underallocation.
In \secref{aligned}, we establish \thmref{unit-size}, assuming the windows are aligned and that $m=1$.
Finally, in \secref{unaligned}, we remove the alignment
assumption from \secref{aligned}, again sacrificing a
constant-factor in the underallocation.

The crux of our new approach to scheduling appears in
\secref{aligned}.  This section gives a simple scheduling policy that
is robust to changes in the scheduling instances.  By contrast, most
classical scheduling algorithms are brittle, where small changes to a
scheduling instance can lead to a cascade of job reallocations even
when the system is highly underallocated.  This brittleness is
certainly inherent to earliest-deadline-first (EDF) and
least-laxity-first (LLF) scheduling policies, the classical greedy
algorithms for scheduling with arrival times and deadlines.  In fact,
we originally expected that any greedy approach would necessarily be
fragile.  We show that this is not the case.

Our new scheduler is based upon a simple greedy policy
(``reser\-vation-based pecking-order scheduling'').  Unlike most robust
algorithms, which explicitly engineer redundancy, the resiliency of
our scheduler derives from a basic combinatorial property of the
underlying ``reservation'' system.   In this sense, it feels different
from typical mechanisms for achieving robustness in computer science
or operations research.

\subsection*{Related Work}  Here, we flesh out the details of related
scheduling and resource allocation work.

\sloppy

  Robust scheduling (or ``robust planning'') involves designing
  schedules that can tolerate some level of
  uncertainty.  See~\cite{KouvelisYu97,MulveyVaZe95} for surveys
  and~\cite{ChiraphadhanakulBa11,JiangBa09,LanClBa06,CapraraGaKr10}
  for applications to train and airline scheduling.
    The assumption in these papers is that the problem
  is approximately static, but there is some error or
  uncertainty, or that the schedule remains near optimal even if the
  underlying situation changes~\cite{Tovey86}.
  By contrast, we focus on an
  arbitrary, worst-case, sequence of requests that may lead to
  significant changes in the overall allocation of resources.
  
  Researchers have also focused on
  finding a good fall-back plan (``reoptimization'') when a schedule is forced to
  change.
  Given an optimal solution for an input, the goal is to compute a
  near-optimal solution to a closely related
  input~\cite{AusielloBoEs07,AusielloEsMo09,ArchettiBeSp03,BockenhauerFoHr06}.
  These papers typically focus on the computational complexity of
  incremental optimization.
  By contrast, we focus on the cost of changing the schedule.

  Shachnai et al.~\cite{ShachnaiTaTa12} introduced a framework that is
  most closely related to ours. They considered computationally
  intractable problems that admit approximation algorithms.  When the
  problem instance changes, they would like to change the solution as
  little as possible in order to reestablish a desired approximation
  ratio.  One difference between their framework and ours is that we
  measure the ratio of reallocation cost to allocation cost, whereas
  there is no notion of initial cost for them. Rather they measure the
  ratio of the transition cost to the optimal possible transition cost
  that will result in a good solution.  Although their framework is an
  analogous framework for approximation algorithms, the particulars
  end up being quite different.

Davis et al.~\cite{DavisEdIm06} propose a
resource reallocation problem where the allocator must assign resources with
respect to a user-determined set of constraints.  The constraints may
change, but the allocator is only informed when the solution becomes
infeasible.  The goals is to minimize communication between the
allocator and the users.  

Many other papers in the literature work within similar setting of
job reallocations, but with different goals, restrictions, or
scheduling problems in mind.  Unal et al.~\cite{UnalUzKi97} study a
problem wherein an initial feasible schedule consisting of jobs with
deadlines must be augmented to include a set of newly added jobs,
minimizing some objective function on only the new jobs without
violating any deadline constraints on the initial schedule.  As in the
present paper they observe that slackness in the original schedule
facilitates a more robust schedule, but outside of the hard
constraints they do not count the reallocation cost.  Hall and
Potts~\cite{HallPo04} allow a sequence of updates and aim to restrict the
change in the schedule, but they evaluate the quality of their
algorithm incrementally rather than with respect to a full sequence of
updates or an offline objective.  

More closely related to our
setting, Westbrook~\cite{Westbrook00} considers the total cost of
migrating jobs across machines in an online load-balancing problem
while also keeping the maximum machine load competitive with the
current offline optimum, which is a different scheduling problem in a
similar framework.  Unlike in the present paper, Westbrook considers
only migration costs and does not include the reallocation cost of
reordering jobs on machines.  Sanders et al.~\cite{SandersSiSk09}
consider a similar load-balancing problem with migration costs and no
reallocation costs; their goal is to study the tradeoff between
migration costs and the instantaneous competitive ratio.



\section{Reallocation Model}
\seclabel{model}

\sloppy

Formally, an \emph{on-line execution} consists of a sequence of
scheduling requests of the following form: $\ang{\proc{InsertJob},
  \id{name}, \id{arrival}, \id{deadline}}$ and $\ang{\proc{DeleteJob},
  \id{name}}$.  A job $j$ has integral arrival time $a_j$ and deadline
$d_j>a_j$, meaning that it must be scheduled in a timeslot no earlier
than time $a_j$ and no later than time $d_j$.  We call the time
interval $[a_j,d_j]$ the job's \defn{window} $W$.  We call $d_j-a_j$,
denoted by $\card{W}$, the \defn{window $W$'s span}.  We use \defn{job
  $j$'s span} as a shorthand for its window's span.  Each job takes
exactly one unit of time to execute.  

At each step, we say that the \defn{active jobs} are those that have
already been inserted, but have not yet been deleted.  Before each
scheduling request, the scheduler must output a feasible schedule for
all the active jobs.  A feasible schedule is one in which each job is
properly scheduled on a particular machine for a time in the the job's
available window, and no two jobs on the same machine are scheduled
for the same time.  Notice that we are not concerned with actually
\emph{running} the schedule; rather, we construct a sequence of
schedules subject to an on-line sequence of requests.

We define the \defn{migration cost} of a request $r_i$ to be the
number of jobs whose machine changes when $r_i$ is processed.  We
define the \defn{reallocation cost} of a request $r_i$ to be the
number of jobs that must be rescheduled when $r_i$ is processed.

When the scheduling instances do not have enough ``slack'' it may
become impossible to achieve low reallocation costs.  In fact, if
there are $n$ jobs currently scheduled, a new request may have
reallocation cost $\Theta(n)$.  Even worse, it may be that most
reallocations require most jobs to be moved, as is shown in
\lemref{why-underallocation}: for large-enough $s$, there exist
length-$s$ request sequences, in which $\Theta(s^2)$ reallocations are
necessary.  Moreover, for large-enough $s$, there exist length-$s$
request sequences in which $\Theta(s)$ machine migrations are
necessary (see \lemref{why-migrations}).

\subsection*{Underallocated Schedules and Our Result} 

To cope with \lemreftwo{why-migrations}{why-underallocation}, we
consider schedules that contain sufficient slack, i.e., that are not
fully subscribed.  We say that a set of jobs is \defn{$m$-machine
  $\gamma$-underallocated}, for $\gamma \geq 1$, if there is a
feasible schedule for those jobs on $m$ machines even when 
the job length (processing time) is multiplied by $\gamma$.
This is equivalent to giving the offline scheduler a processing speed
that is $\gamma$ times slower than the online scheduler.  When $m$ is
implied by context, we simply say \defn{$\gamma$-underallocated}.

Overloading terminology, we say that a sequence of scheduling requests is
\defn{$\gamma$-underallocated} if after each request the set of active jobs is
$\gamma$-underallocated.

\subsection*{Aligned-Windows Assumption} 

The assumption of aligned windows is used in
\secreftwo{parallel}{aligned}, but it is dropped in \secref{unaligned}
to prove the full theorem.  We say that a window $W$ is \defn{aligned}
if (i) it has span $2^i$, for some integer $i$, and (ii) it has a
starting time that is a multiple of $2^i$.  If a job's window is
aligned, we say that the job is \defn{aligned}.  We say that a set of
windows (or jobs) are \defn{recursively aligned} if every window (or
job) is aligned.

Notice that recursive alignment implies that two jobs windows are
either equal, disjoint, or one is contained in the other (i.e., the
windows are laminar).  Dealing with recursively aligned windows is
convenient in part due to the following observation.

\begin{lemma}\lemlabel{underallocated}
  If a recursively aligned set of jobs is $m$-machine
  $\gamma$-underallocated, then for any aligned window $W$ there are
  at most $m\card{W}/\gamma$ jobs with span at most $\card{W}$ whose
  windows overlap $W$.
\end{lemma}

\begin{proof}
The window $W$ comprises $\card{W}$ timeslots on each
  of $m$ machines, for a total of $m\card{W}$ timeslots.  By
  definition, a $\gamma$-underallocated instance is feasible even if
  the jobs' processing times are increased to $\gamma$.  Thus, there
  may be at most $m\card{W}/\gamma$ jobs restricted to window $W$.
  Since the set of jobs is recursively aligned, if a job has window
  $W'$ that overlaps $W$ and $\card{W'} \leq \card{W}$, then $W'$ is
  fully contained by $W$.  Hence, there can be at most
  $m\card{W}/\gamma$ such jobs.
\end{proof} 


\section{Reallocating Aligned Jobs on Multiple Machines}
\seclabel{parallel} 

This section algorithmically reduces the multiple-machine scheduling
problem to a single-machine scheduling problem, assuming recursive
alignment.  The reduction uses at most one migration per request.   We
use $m$ to denote the number of machines.

The algorithm is as follows.  For every window $W$, record the number
$n_W$ of jobs having window $W$. (This number need only be recorded
for windows that exist in the current instance, so there can be at
most $n$ relevant windows for $n$ jobs.) The goal is to maintain the
invariant that every machine has between $\floor{n_W/m}$ and
$\ceil{n_W/m}$ jobs with window $W$, with the extra jobs being
assigned to the earliest machines.  This invariant can be maintained
simply by delegating jobs, for each window $W$, round-robin: if there
are $n_W$ jobs with window $W$, a new job with window $W$ is delegated
to machine $(n_W + 1) \bmod m$.  When a job with window $W$ is deleted
from some machine $m_i$, then a job is removed from machine $(n_w \bmod
m)$ and migrated to machine $m_i$.  All job movements are performed
via delegation to the single-machine scheduler on the specified machine(s).

The remaining question is whether the instances assigned to each
machine are feasible.  The following lemma says that they are.

\begin{lemma}
\lemlabel{parallel}
  Consider any $m$-machine $6\gamma$-underallocated recursively
  aligned set of jobs $J$, where $\gamma$ is an integer.  Consider a
  subset of jobs $J'$ such that if $J$ contains $n_W$ jobs of window
  $W$, then $J'$ contains at most $\ceil{n_W/m}$ jobs of window $W$.  Then
  $J'$ is 1-machine $\gamma$-underallocated.
\end{lemma}
\begin{proof}
  Since $J$ is underallocated, \lemref{underallocated} says that there
  can be at most $m\card{W}/(6\gamma)$ jobs with window $W$
  or nested inside $W$. By definition, no window smaller than
  $6\gamma$ contains any jobs.  The worry is that the ceilings add too
  many jobs to one machine.  But there are at most
  $2\card{W}/(6\gamma)$ windows nested inside $W$, and the ceilings add
  at most 1 job to each of these windows.   So the total number of
  jobs in $J'$ with windows inside $W$ is at most $\card{W}/(6\gamma)
  + 2\card{W}/(6\gamma) = \card{W}/(2\gamma)$.  Even if all jobs are
  restricted to run at multiples of $\gamma$, a simple inductive
  argument shows that this many size-$\gamma$ jobs can be feasibly scheduled.
\end{proof}



\section{Reallocating Aligned Jobs on One Machine}
\seclabel{aligned}

We now give a single machine, reallocating scheduler for unit-sized
jobs.  We assume a bound $n$ on the number of jobs concurrently
scheduled in the system, and relax this assumption at the end of the
section.

\subsection*{Na\"{\i}ve  Pecking-Order Scheduling is\\ Logarithmic}

We first give the na\"{\i}ve solution, which requires a logarithmic
number of reallocations per job insert/delete.
This solution uses what we call \defn{pecking-order scheduling}, which
means that a job $k$ schedules itself without regard for jobs with
longer span and with complete deference to jobs with shorter span.  A
job $k$ with window $W$ may get displaced by a job $j$ with a shorter
window (nested inside $W$), and $k$ may subsequently displace a job
$\ell$ with longer window.\footnote{At first glance, \lemref{naive}
  seems to contradict the underallocation requirement given in
  \lemref{why-underallocation}.  That lower bound, however, applies to
  the general case, whereas this lemma applies to the aligned case.}

\begin{lemma}
  Let $n$ denote the maximum number of jobs in any schedule and let
  $\Delta$ denote the longest window span.  There exists a greedy
  reallocating scheduler such that for every feasible sequence of
  recursively aligned scheduling requests, the reallocation cost of
  each insert/delete is $O(\min\left\{\log n,\log \Delta\right\})$.
  \lemlabel{naive}
\end{lemma}
\begin{proof}
  To insert a job $j$ with span $2^i$, find any empty slot in $j$'s
  window, and place $j$ there.  Otherwise, select any job $k$
  currently scheduled in $j$'s window that has span $\geq 2^{i+1}$.
  If no such $k$ exists, the instance is not feasible (as every job
  currently scheduled in $j$'s window \emph{must} be scheduled in
  $j$'s window).  If such a $k$ exists, replace $k$ with $j$ and
  recursively insert $k$.  This strategy causes cascading reallocations
  through increasing window spans, reallocating at most one job with
  each span.  Since there are at most $\log\Delta$ distinct window
  spans in the aligned case, and moreover all jobs can fit within a
  window of span $n$, the number of cascading reallocations is
  $O(\min\left\{\log n, \log\Delta\right\})$. 
\end{proof}

\subsection*{Pecking-Order Reallocation via Reservations Costs
\boldmath {\large $O(\min \{\log^* n, \log^* \Delta\})$}}

We now give a more efficient reallocating scheduler, which matches
\thmref{unit-size} when the scheduling requests are recursively
aligned.  The algorithm is summarized for job insertions in
\figref{algorithm}.  

The intuition behind reservation scheduling manifests itself in the
process of securing a reservation at a popular restaurant.  If
higher-priority diners already have reservations, then our reservation
is waitlisted. Even if our reservation is ``confirmed,'' a celebrity
(or the President, for DC residents) may drop in at the last moment
and steal our slot.  If the restaurant is empty, or full of
low-priority people like graduate students, then our reservation is
fulfilled.  The trick to booking a reservation at a competitive
restaurant is to make several reservations in parallel.  If multiple
restaurants grant the reservation, we can select one to eat at.  If a
late arrival steals our slot, no problem, we have another reservation
waiting.

Back to our scheduling problem, by spreading out reservations
carefully, jobs will only interfere if they have drastically different
spans.  Our algorithm handles jobs with ``long'' windows and ``short''
windows separately, and only a ``short'' job can displace a long job.
The scheduler itself is recursive, so ``very short'' jobs can displace
``short'' jobs which can displace ``long'' jobs, but the number of
levels of recursion here will be $\log^*\Delta$, as opposed to $\log
\Delta$ in the naive solution.

There are two components to the scheduler. The first component uses
reservations to guarantee that jobs cannot displace (many) other jobs
having ``similar'' span, so the reallocation cost if all jobs have
similar spans is $O(1)$.  These (over-)reservations, however, consume
timeslots and amplify the underallocation requirements.  Applying the
scheduler recursively at this point is trivial to achieve a good
reallocation cost, but the required underallocation would become
nonconstant. The second component of the scheduler is to combine
levels of granularity so that their effects on underallocation do not
compound.

The remainder of the section is organized as follows.  We first
discuss an interval decomposition to separate jobs into different
``levels'' according to their spans.  Then we present the scheduler
with regards to a single job level. Finally we discuss how to
incorporate multiple levels simultaneously.

\subsubsection*{Interval Decomposition}

Our scheduler operates nearly independently at multiple levels of
granularity.  More precisely, we view these levels from bottom up by
defining the threshold
\[ L_{\ell+1} = 
      \begin{cases}
        2^{5} & \text{if $\ell = 0$} \\
        2^{L_{\ell}/4} & \text{ if $\ell > 0$}  
      \end{cases} \ .
\]
It is not hard to see that $L$ is always a power of 2, growing as a
tower function of $\sqrt[4]{2}$.  It is often convenient to use the
equivalent relationship $L_{\ell} = 4\lg(L_{\ell+1})$---each threshold
is roughly the $\lg$ of the next.

Our scheduler operates recursively according to these thresholds.  The
level-$\ell$ scheduler handles jobs and windows $W$ with span
$L_{\ell} < \card{W} \leq L_{\ell+1}$.  We call a job (or window) a
\defn{level-$\ell$ job (window)} if its span falls in this range.

We partition level-$\ell$ windows into nonoverlapping, aligned
subwindows called \defn{level-$\ell$ intervals}, consisting of $L_\ell
= 4\lg L_{\ell+1}$ timeslots.  The following observation is useful in
our analysis:
\begin{equation} \left(\text{\# of distinct level-$\ell$-window spans}\right)
  \leq \lg(L_{\ell+1}) = 
L_\ell/4 \label{eqn:windowcount}
\end{equation}

The reallocation scheduler operates recursively within each interval
to handle lower-level jobs. Because this is pecking-order scheduling,
the recursive scheduler makes decisions without paying attention to
the location of the higher-level jobs, guaranteeing only that each
lower-level job is assigned a unique slot within its appropriate
window.  In doing so, it may displace a long job and invoke the
higher-level scheduler.

\sloppy 

\subsubsection*{Schedule Level-{\large $\ell$} Jobs via Reservations}

Consider a level-$\ell$ window $W$ with span $2^k L_\ell$, for some
integer $k \geq 1$ (i.e., $W$ contains $2^k$ level-$\ell$ intervals).
Let $x$ denote the number of jobs having exactly window $W$. 

The window $W$ maintains a set of \defn{reservations} for these $x$
jobs, where each reservation is a \emph{request for a slot in a given
  level-$\ell$ interval}.
A reservation made by $W$ can be \defn{fulfilled}; this means that one
slot from the requested interval is \defn{assigned to $W$}, and the
only level-$\ell$ jobs that may \defn{occupy} that slot are any of the
$x$ jobs with window exactly $W$.
Alternatively, a reservation can be \defn{waitlisted}; this means that
all the slots in the requested interval are already assigned to
smaller windows than $W$.  Which reservations are fulfilled and which
are waitlisted may change over time as jobs get allocated and removed.

We now explain how these reservations are made.  Initially, a
level-$\ell$ window $W$ makes one reservation for each enclosed
level-$\ell$ interval.  It makes two additional reservations for each
job having window $W$. These reservations are spread out round-robin
among the intervals within $W$ (and independently of any jobs with any
different windows).  We maintain the following invariant:

\begin{invariant}
  If there are $x$ jobs having level-$\ell$ window $W$ with
  $\card{W}=2^kL_\ell$, then $W$ has exactly $2x+2^k$ reservations in
  level-$\ell$ intervals.
\begin{itemize}
\item These reservations are assigned in round-robin order to the
  intervals in $W$.
\item Each of the enclosed intervals contains either
  $\floor{2x/2^k}+1$ or $\floor{2x/2^k}+2$ of $W$'s reservations,
  where the leftmost intervals have the most reservations and the
  rightmost intervals have the least reservations.
  \invlabel{reservations}
\end{itemize}
\end{invariant}

To maintain \invref{reservations}, when a new job with window $W$ is
allocated, $W$ makes two new reservations, and these are sent to the
leftmost intervals that have the least number ($\floor{2x/2^k}+1$) of
$W$'s reservations.  When a job having window $W$ is deleted, $W$
removes one reservation each from the two rightmost intervals that
have the most reservations.

We now describe the reservation process from the perspective of the
interval, which handles reservation requests from the $< L_\ell/4$
level-$\ell$ windows that contain the interval (see
Equation~\ref{eqn:windowcount}).  The interval decides whether to
fulfill or waitlist a reservation, prioritizing reservations made by
shorter windows.  Each interval $I$ has an \defn{allowance}
$\id{allowance}(I)$, specifying which slots it may use to fulfill
reservations. In the absence of lower-level jobs, the
$\card{\id{allowance}(I)} = L_{\ell}$, since the interval has span
$L_{\ell}$.  (When lower-level jobs are introduced, however, the
allowance decreases---the allowance contains all those slots that are
not \emph{occupied} by lower-level jobs.)  Thus, the interval sorts the
window reservations with respect to span from shortest to longest, and
fulfills the $\card{\id{allowance}(I)} \leq L_{\ell}$ reservations
that originate from the shortest windows.  A fulfilled reservation is
assigned to a specific slot in the interval, while a waitlisted
reservation has no slot.  The interval maintains a list of these
waitlisted reservations.

The set of fulfilled reservations changes dynamically as
insertions/deletions occur.  When a new reservation is made by window
$W$, a longer window $W'$ may lose a reserved slot as one of its
fulfilled reservations is moved to the waitlist; if there is a job (of
the same level) in that slot, it must be moved.  When a job with
window $W$ is deleted, $W$ has two fewer reservations, and so may lose
two fulfilled slots.  If there is a job in either of these slots, then
that job must be moved.  (In this case, a longer window $W'$ may gain
a fulfilled slot, but this does not require any job movement.)
The following invariant is needed to establish the algorithm's
correctness.

\begin{invariant}
  When a job having window $W$ is newly allocated, $W$ makes two new
  reservations. Then the job is assigned to any empty slot for which
  $W$ has a fulfilled reservation.  There will always be at least one
  such slot (proved by \lemref{reservationspace}).
  \invlabel{slot-always-exists}
\end{invariant}

Interestingly, as a consequence of
pecking-order scheduling combined with round-robin reservations:
\begin{observation}
  Which reservations in which intervals are fulfilled and which are
  waitlisted is history independent.  The actual placement of the jobs
  is not history independent.  \obslabel{history-independent}
\end{observation}

\setlength{\codeboxwidth}{.8\textwidth} 
\addtolength{\codeboxwidth}{-2\digitwidth}

\begin{figure*}[t]
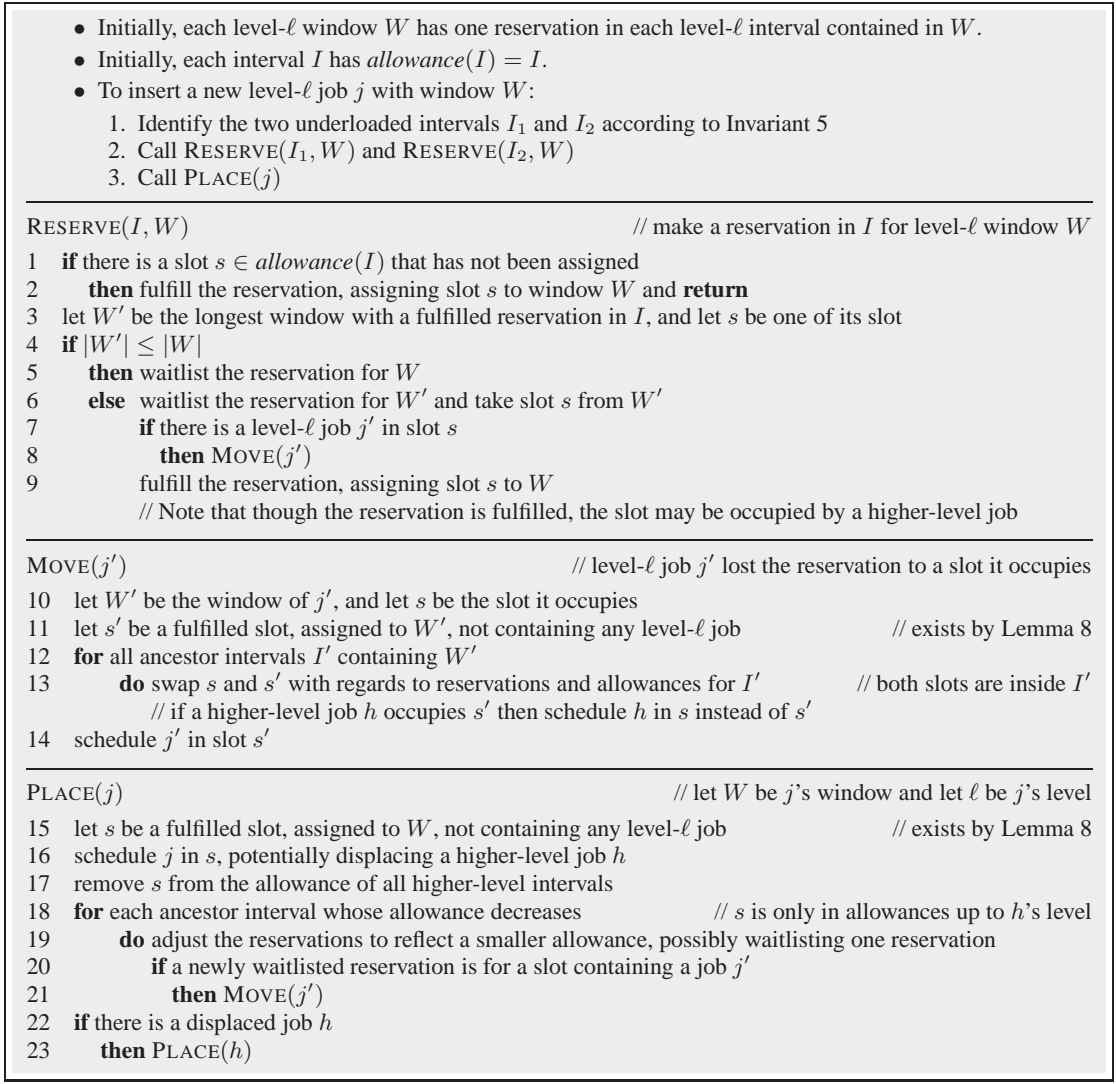
\figlabel{algorithm}
\begin{center}
\fbox{\colorbox{light-gray}{
\begin{minipage}[t]{.8\textwidth}
  \begin{itemize}
  \item Initially, each level-$\ell$ window $W$ has one
    reservation in each level-$\ell$ interval contained in $W$.
  \item Initially, each interval $I$ has $\id{allowance}(I) = I$.
  \item To insert a new level-$\ell$ job $j$ with window $W$:
    \begin{enumerate}
    \item Identify the two underloaded intervals $I_1$ and $I_2$
      according to \invref{reservations}
    \item Call $\proc{Reserve}(I_1,W)$ and $\proc{Reserve}(I_2,W)$
    \item Call $\proc{Place}(j)$
    \end{enumerate}
  \end{itemize}
\vspace{-.2em}\hrule\vspace{-.7em}
  \begin{codebox}
    \Procname{$\proc{Reserve}(I,W)$ \hfill \Comment{make a reservation
        in $I$ for level-$\ell$ window $W$}} \li \If there is a slot
    $s\in\id{allowance}(I)$ that has not been assigned \li \Then fulfill the reservation, assigning slot $s$ to window $W$
    and \Return
        \End
    \li let $W'$ be the longest window with a fulfilled
    reservation in $I$, and let $s$ be one of its slot
    \li \If $\card{W'} \leq \card{W}$
    \li \Then waitlist the reservation for $W$
    \li \Else waitlist the reservation for $W'$ and take slot
    $s$ from $W'$
    \li       \If there is a level-$\ell$ job $j'$ in slot $s$
    \li       \Then $\proc{Move}(j')$
              \End
    \li       fulfill the reservation, assigning slot $s$ to $W$
    \zi       \Comment{Note that though the reservation is fulfilled, the slot
    may be occupied by a higher-level job}
        \End
  \end{codebox}
\vspace{-.7em}\hrule\vspace{-.7em}
  \begin{codebox}
    \setcounter{codelinenumber}{9}
    \Procname{$\proc{Move}(j')$ \hfill \Comment{level-$\ell$ job $j'$ lost the
        reservation to a slot it occupies}}
    \li let $W'$ be the window of $j'$, and let $s$ be the slot it occupies
    \li let $s'$ be a fulfilled slot, assigned to $W'$, not containing
    any level-$\ell$ job\RComment{exists by \lemref{reservationspace}}
    \li \For all ancestor intervals $I'$ containing $W'$
    \li \Do swap $s$ and $s'$ with regards to reservations and
    allowances for $I'$ \RComment{both slots are inside $I'$}
    \zi     \Comment{if a higher-level job $h$ occupies $s'$ then
      schedule $h$ in $s$ instead of $s'$}
        \End
    \li schedule $j'$ in slot $s'$    
  \end{codebox}
\vspace{-.7em}\hrule\vspace{-.7em}
  \begin{codebox}
    \setcounter{codelinenumber}{14}
    \Procname{$\proc{Place}(j)$ \hfill \Comment{let $W$ be $j$'s
      window and let $\ell$ be $j$'s level}}
    \li let $s$ be a fulfilled slot, assigned to $W$, not containing
    any level-$\ell$ job \RComment{exists by
      \lemref{reservationspace}}
    \li schedule $j$ in $s$, potentially displacing a higher-level job
    $h$
    \li remove $s$ from the allowance of all higher-level intervals
    \li \For each ancestor interval whose allowance decreases
    \RComment{$s$ is only in allowances up to $h$'s level}
    \li \Do adjust the reservations to reflect a smaller allowance,
    possibly waitlisting one reservation
    \li     \If a newly waitlisted reservation is for a slot containing a job $j'$
    \li     \Then $\proc{Move}(j')$
            \End
        \End
    \li \If there is a displaced job $h$
    \li \Then $\proc{Place}(h)$
        \End
  \end{codebox}
\end{minipage}
}}
\end{center}
\caption{Pecking-order scheduling with reservations.}
\end{figure*}

\subsubsection*{Scheduling Across All Levels}

Consider inserting a level-$\ell$ job $j$.  Suppose $j$'s window is
contained in a higher-level interval $I'$.  We schedule $j$ at its own
level according to the pecking-order scheduler, without regard to
higher-level schedulers.  Recall that the first step of the insertion
is placing two new reservations. Whenever the reservations cause
another level-$\ell$ job $j'$ to move from slot $s$ to slot $s'$, the
allowance of all higher-level intervals must be updated to
reflect the change in slot usage. However, since both $s \in I'$ and
$s' \in I'$, and $j'$ vacates the original slot $s$, there is no net
change to $\card{\id{allowance}(I')}$.  It is thus sufficient to swap
$s$ and $s'$ for all higher-level intervals $I'$, which may result in
a total of one higher-level job movement.

After updating the reservations, the new job $j$ is placed in one of
its assigned slots $s$.  This slot may either be empty, or it may
contain a higher-level job $h$---the scheduler chooses $s$ without
regard to these possibilities.  In either case, the slot $s$ will be
used by $j$, so it must be removed from $\id{allowance}(I')$ for any
ancester interval $I'$---meaning the higher-level scheduler cannot use
this slot.  If the slot $s$ was empty, then the job $j$ is assigned to
that slot and the insertion terminates.  If the slot $s$ was
previously occupied by a higher-level job $h$, then $h$ is displaced
and a new slot must be found.  Unlike in the case of reservations,
$\card{\id{allowance}(I')}$ decreases here and we do not immediately
have a candidate slot into which to place $h$.  Instead, we reinsert
$h$ recursively using the scheduler at its level.  This displacement
and reinsertion may cascade to higher levels.

Observe that the higher-level scheduler is unaware of the reservation
system employed by the lower-level scheduler.  It only knows which
slots are in its allowance. These slots are exactly those that are not
\emph{occupied} by short-window jobs.  The interval does not observe
the reservations occurring within nested intervals---only actual job
placement matters.  When a lower-level job is deleted, the allowance
of the containing interval increases to include the slot that is no
longer occupied.

\subsection*{Reservation Analysis}

We now use the following lemma to establish
\invref{slot-always-exists}, which claims that there are always enough
fulfilled reservations.  Since the reservations fulfilled by each
interval are history independent (see \obsref{history-independent}),
this proof applies at all points during the execution of the
algorithm.

\begin{lemma}
  Suppose that a sequence of aligned scheduling requests
  8-underallocated.  If there are $x$ jobs each having the same window
  $W$, then $W$ has at least $x+1$ fulfilled
  reservations.\lemlabel{reservationspace}
\end{lemma}

\begin{proof}
  Let $\card{W} = 2^kL_\ell$ for level-$\ell$ window $W$.  Let $y$ be
  the number of level-$\ell$ jobs with windows nested inside $W$.
  Each of those windows makes 2 reservations for each job, plus an
  extra reservation to each of the $2^k$ intervals.  So the total
  number of reservations in $W$ is at most $2(x+y) + 2^k\lg W$. In
  addition, let $z$ be the number of lower-level jobs nested inside
  $W$.  Since we are 8-underallocated, we have $2(x+y) + z \leq
  2(x+y+z) \leq \card{W}/4$ by
  \lemref{underallocated}. By Equation~\ref{eqn:windowcount}, we have $\lg W
  \leq L_\ell/4$, and hence $2^k\lg W \leq (2^kL_\ell)/4 =
  \card{W}/4$.  Summing these up, we have that at most $\card{W}/2$
  slots consumed by lower-level jobs and these reservations.

  In order for a particular interval to waitlist even one of $W$'s
  reservation requests, it would need to have strictly more than
  $L_\ell$ of these reservations or lower-level jobs assigned to
  it. But there are only $\card{W}/2$ slots consumed in total, so
  strictly less than $1/2$ the intervals can waitlist even one of
  $W$'s reservations.  Since window $W$ reserves at least
  $\floor{2x/2^k}+1$ slots in every one of the $2^k$ intervals by
  \invref{reservations}, it must therefore be granted strictly more
  than $(\floor{2x/2^k}+1)(1/2)(2^k) \geq x$ fulfilled reservations.
\end{proof}

Since each window $W$ containing $x$ jobs has at least $x+1$ fulfilled
reservations at intervals within $W$, there is always an appropriate
slot to schedule a new belonging to this window.  This ensures that
there each operation leads to only $O(1)$ reallocations at each level.

\subsection*{\boldmath Trimming Windows to {\large $n$} and Deamortization}

Ideally, the reallocation cost for a request $r$ should be a function
of the number of active jobs $n_r$ in the system when request $r$ is
made.
To achieve this performance guarantee, we maintain a value $n^*$ that
is roughly the number of jobs in the current schedule.  When the number of
active jobs exceeds $n^*$, we double $n^*$; when the number of active
jobs shrinks below $n^*/4$, we halve $n^*$.

For every job that has a window larger than $2\gamma n^*$, we trim its
window---reducing it arbitrarily to size $2\gamma n^*$.  The adjusted
instance remains $\gamma$-underallocated, since there are at most $n^*$
other jobs scheduled in the trimmed window of size $2\gamma n^*$.

To achieve good amortized performance, it is enough to rebuild the
schedule from scratch each time we change the value of $n^*$.  This
rebuilding incurs an amortized $O(1)$ reallocation cost.

This amortized solution can be deamortized, as long as the scheduling
instance is sufficiently underallocated that the following property
holds: if each job is duplicated (i.e., inserted twice on inserts,
deleted twice on delete), the resulting instance is
$\gamma$-underallocated, for appropriate constant $\gamma$.  This
property holds as long as the initial (unduplicated) scheduling
instance is $2\gamma$-underallocated.

The idea is to rebuild the schedule gradually, performing a little
update every time a new reallocation request is serviced.  This
approach is reminiscent to how one deamortizes the rebuilding of a
hash table that is too full or too empty.
We use the even (or odd) time slots for the old schedule and the odd
(or even) time slots for the new schedule.  Instead of rebuilding the
schedule all at once, every time one job is added or deleted, two jobs
are moved from the old schedule to the new schedule. 

\subsection*{Wrapping Up}

We conclude with the following lemma, which puts together the various
results in this section:
\begin{lemma}
  There exists a constant $\gamma$ and a single-machine reallocating
  scheduler such that for any $1$-machine $\gamma$-underallocated
  sequence of aligned scheduling requests, we achieve the following
  performance.  Let $n_i$ denote the number of jobs in the schedule
  and $\Delta_i$ the largest window size when the $i$th reallocation
  takes place.  Then the $i$th reallocation has cost
  $O(\min\left\{\log^* n_{i},\log^* \Delta_{i}\right\})$.
  \lemlabel{singlemachine}
\end{lemma}
\begin{proof}
We consider the performance of the pecking-order scheduler with
reservations, where we maintain an estimate $n^*$ via deamortized
shrinking and doubling and trim all windows to $\gamma n^*$.  

\lemref{reservationspace} shows that there is always a slot available
to put a job (\invref{slot-always-exists}), and hence we observe that
there are at most $O(1)$ reallocations at each level of the scheduler.
Specifically, on insertion, the two reservations may result in two
calls to \proc{Move} for jobs at the same level as the one being
inserted.  Each \proc{Move} results in one reallocation of the job
being moved, plus at most one reallocation at a higher level.  Then
the call to \proc{Place} may cascade across all levels, but it in
aggregate it only includes one \proc{Move} per level, each causing at
most two reallocations.

If $\Delta_i$ is the largest job size when operation $i$ occurs, there
are no more than $O(\log^*\Delta_i)$ levels.  Since $n_i^* \leq 4n$, and all
windows are trimmed to length $\gamma n^*$, we also know that there
are no more than $O(\log^* (4\gamma n_i))$ levels.  From this the result
follows.
\end{proof}


\section{Reallocating Unaligned Jobs on Multiple Machine}
\seclabel{unaligned}

\renewcommand{\aligned}[1]{\proc{aligned}(\ensuremath{#1})}

In this section, we generalize to jobs that are not aligned, removing
the alignment assumptions that we made in
\secreftwo{parallel}{aligned}.  We show that if $S$ is a
$\gamma$-underallocated sequence of scheduling requests, then we can
safely \emph{trim} each of the windows associated with each of the
jobs, creating an aligned instance.  Since the initial sequence of
scheduling requests is underallocated, the resulting aligned sequence
is also underallocated, losing only a constant factor.

We first define some terminology.  If $W$ is an arbitrary window, we
say that $\aligned{W}$ is a largest aligned window that is contained
in $W$.  (If there is more than one largest window, choose
arbitrarily.) Notice that $|\aligned{W}| \geq \card{W}/4$.  If $J$ is
a set of jobs, then $\aligned{J}$ is the set of jobs in which the
window $W$ associated with each job is replaced with $\aligned{W}$.

\begin{lemma} \lemlabel{aligned}
  Consider any $m$-machine $4\gamma$-underallocated set of jobs $J$,
  where $\gamma$ is an integer.  Then \aligned{J} is
  $m$-machine $\gamma$-underallocated.
\end{lemma}
\begin{proof}
  Assume for the sake of contradiction that $\aligned{J}$ is not
  $\gamma$-underallocated.  This implies that there must exist a
  window $W$ that has $> m\card{W}/\gamma$ jobs with trimmed windows
  contained in $W$ (as otherwise we could schedule the size $\gamma$
  jobs via a simple inductive argument).  Let $J' \subseteq J$ be the
  jobs whose trimmed windows are contained in $W$.

Since $J$ is $4\gamma$-underallocated, we now examine an (unaligned)
scheduling of the jobs in $J'$ that satisfies the
$4\gamma$-underallocation requirement.  We observe that all the jobs
in $J'$ are scheduled in a region of size at most $4\card{W}$.
  However, since the schedule is $4\gamma$-underallocated, there can be
  at most $4m|W|/(4\gamma)$ jobs in this region of size $4|W|$.
  That is $|J'| \leq m|W|/\gamma$, which is a contradiction.
\end{proof}

\subsection*{}

\noindent From this we can  conclude with the proof of
\thmref{unit-size}:

\medskip

\noindent
{\sc{Proof of \thmref{unit-size}.}}
Jobs are scheduled as
  follows: first, a new job has its window aligned; second, it is
  delegated to a machine in round-robin fashion; finally, it is
  scheduled via single-machine pecking-order scheduling with
  reservations.  When a job is deleted, it is removed by the
  appropriate single-machine scheduler, and then there is at most one
  migration to maintain the balance of jobs across machines.  This is
  the only time that jobs migrate.

  \lemref{aligned} shows that the set of aligned jobs is $m$-machine
  $\gamma/4$-underallocated, and \lemref{parallel} shows that the jobs
  asigned to each machine are $1$-machine $\gamma/24$-underallocated.
  Finally, \lemref{singlemachine} shows that each single-machine
  scheduler operation has cost $O(\min\left\{\log^* n_{i},\log^*
    \Delta_{i}\right\})$---and each job addition or deletion invokes
  $O(1)$ single-machine scheduler operations.\putqed



\section{What Happens Without\\ Underallocation?}
\seclabel{why}

This section explains what happens without underallocation and
why migrations are necessary at all. 

If migration cost is to be bounded only by reallocation cost and since 
jobs have unit size, it is
trivial to transform a parallel instance to a single-machine instance
my making a single machine go $m$ times faster. 
Since migration cost
across machines could be more expensive than rescheduling a single
machine, we are interested in providing a tighter bound on the
migration cost.  The question then is: \emph{are migrations
  necessary?}  The following lemma shows that they are.  In
fact, the per-request migration cost must be $\Omega(1)$ in the
worst-case for any deterministic algorithm.

\begin{lemma}
There exists a sufficiently large sequence of $s$
job insertions/deletions
on  $m>1$ machines, such that any
deterministic scheduling algorithm has a total migration cost of
$\Omega(s)$.
\lemlabel{why-migrations}
\end{lemma}
\begin{proof}
  Without loss of generality, assume $6m$ divides $s$. Divide the $s$
  requests into $s/(6m)$ consecutive subsequences of $6m$ requests
  each.  Each subsequence is as follows:
  \begin{closeenum}
  \item Insert $2m$ span-2 jobs with window $[0,2]$.
  \item Delete the $m$ jobs scheduled on the first $m/2$ machines. 
  \item Insert $m$ span-1 jobs with windows $[0,1]$.  
  \item Delete all $2m$ remaining jobs.
  \end{closeenum}
  After step~1, the only feasible schedule is to put two jobs on each
  machine.  After step 2, half the machines have two jobs, and the other
  half of the machines have no jobs.  The only feasible schedule after
  step 3 is to have on each machine a span-1 job starting at time 0, and a
  span-2 job starting at time 1.  This means that half of the span-2
  jobs must migrate across machines, causing $m/2$ migrations.  There
  are thus $m/2$ migrations for every $6m$ requests, or a total of
  $s/12$ migrations.
\end{proof}

It is also easy to see that for some sequences of scheduling requests,
if they are not underallocated, 
it is impossible to achieve low reallocation costs, even if there
exists a feasible schedule. 

\begin{lemma}\lemlabel{expensive}
There exists a sequence of $s$ job inserts/deletions, 
such that any scheduling algorithm has a
rescheduling cost of $\Omega(s^2)$.
  \lemlabel{why-underallocation}
\end{lemma}
\begin{proof}
  Consider for example $\eta=s/2$ jobs numbered $0,1,\ldots,\eta-1$,
  where job $j$ has window $[j,j+2]$.  With the insertion of one
  additional job having window $[0,1]$, forcing the job to be
  scheduled at time 0, all $\eta$ other jobs are forced to schedule
  during their later slot. If that job is deleted and another
  unit-span job with window $[\eta,\eta+1]$ is inserted, then all jobs
  are forced to schedule during their earlier slot. By toggling
  between these two options, all jobs are forced to move,
  resulting in cost $\Omega(\eta)$ to handle each request.  Repeating
  $\eta$ times gives a total cost of $\Omega(\eta^2) = \Omega(s^2)$.  
\end{proof}


\section{Conclusions and Open\\ Questions}
\seclabel{conclusion}
  
The results in this paper suggest several followup questions.
First, is it possible to generalize this paper's 
reallocation scheduler for the case where jobs are not unit-sized?  
Observe that we are 
limited by the computational difficulty of scheduling with arrival times and 
deadlines
when jobs are not unit size; see~\cite{BansalChKh07} for recent 
results with resource augmentation. 
We are also limited by the following 
observation: 
\begin{observation}
  Suppose there exist jobs of size $1$ and jobs of size $k$, for any $k>1$.  For
  any reallocation scheduler, there is a sequence of $\Theta(n)$
  scheduling requests that has aggregate reallocation cost
  $\Omega(kn)$, for $k \leq n$, even if the requests are
  $\gamma$-underallocated for any constant $\gamma$.
\obslabel{size-k-jobs}
\end{observation}
\begin{proof}
  Consider a schedule of length $m = 2 \gamma k$.  Assume there are
  $k$ unit-sized jobs that are each scheduled with a window beginning
  at $0$ and ending at $m$.  In addition, consider a single large job
  $p$ that has size $k$ and a window of span exactly $k$.  

  Initially, all $k$ unit-size jobs are scheduled and they remain in
  the system throughout.  The large job $p$ is initially scheduled at
  time slot 0.  It is then deleted from time slot 0 and re-inserted at
  time slot $k$, and then again at time slot $2k, 3k, \ldots, m-k$.
  The same sequence of $2\gamma$ insertions and deletions is then
  repeated $n$ times.

  During a single sequence of $2\gamma$ insertions and deletions, each
  of the $k$ unit-sized jobs has to be rescheduled at least once,
  resulting in $\Omega(kn)$ reallocation cost.
\end{proof}

Does there exist a reallocation scheduler that handles jobs 
whose sizes are integers less than or equal to $k$ and 
matching the bounds in \obsref{size-k-jobs}?
There could be applications where jobs are not unit size, but 
where $k$ is relatively small. 

What happens if other types of reallocations are allowed, such as 
if new machines can be added or dropped from the schedule, 
or if machine speeds can change?

In this paper, $\gamma$ is very large, 
and the  paper does not attempt to optimize this constant, 
preferring clarity of exposition.
How much can this constant be improved? Is there a reallocation scheduler 
where $\gamma=1+\varepsilon$?

Finally, 
what other scheduling and optimization problems lend themselves to 
study in the context of reallocation?

\bibliographystyle{abbrv}
\bibliography{lit,lit-more}

\end{document}